\documentclass[12pt,a4paper,twoside,notitlepage]{article}
\usepackage[english]{babel}
\usepackage[T1]{fontenc} 
\usepackage{graphicx}
\usepackage{amsmath,amsthm,epsfig,amsfonts,bbm} 

\DeclareMathOperator{\sgn}{sgn}

\newcommand{\Frac}[2] {\frac{\mbox{\normalsize{$#1$}}}{\mbox{\normalsize{$#2$}}}}
\newcommand{\DP}[2]{\displaystyle \frac{\partial #1}{\partial #2}}

\newcommand{\pequationdeb}{$$ \left\{ \begin{minipage}[c]{130mm}}
\newcommand{\pequationfin}{\end{minipage}
                           \right. $$}
\newcommand{\vite}{\mathbf{u}}
\newcommand{\srb}{\sqrt{R}}
\newcommand{\Rb}{R}
\newcommand{\beq}     {\begin{equation}}
\newcommand{\enq}     {\end{equation}}
\newcommand{\be}    {\begin{enumerate}}
\newcommand{\ee}    {\end{enumerate}}

\newcommand{\Bb}

\newcommand{\RE}{\mbox{{\rm Re}}}
\newcommand{\Bo}{\mbox{{\rm Bo}}}
\newcommand{\sqb}{\sqrt{\varepsilon}}

\newcommand{\Int}     {\displaystyle \int}

\newtheorem{theorem}{Theorem}
\newtheorem{lemma}[theorem]{Lemma}
\newtheorem{remark}[theorem]{{\em Remark}}
\newtheorem{proposition}[theorem]{Proposition}

\begin{document}
\thispagestyle{empty}
\setcounter{page}{1}

\title{Derivation of a viscous KP equation including surface tension, and related equations}

\author{Herv\'e V.J. Le Meur$^{1}$\thanks{The author wants to thank Luc Molinet and Jean-Claude Saut for fruitful discussions.}\\
$^1$ Laboratoire de Math\'ematiques d'Orsay, Univ. Paris-Sud, \\
CNRS, Universit\'e Paris-Saclay, 91405 Orsay, France. \\
\texttt{Herve.LeMeur@math.u-psud.fr}}

\date{October 2015}


\maketitle

\begin{abstract} The aim of this article is to derive surface wave
  models in the presence of surface tension and viscosity. Using the
  Navier-Stokes equations with a free surface, flat bottom and surface
  tension, we derive the viscous 2D Boussinesq system with a weak
  transverse variation. The assumed transverse variation is on a
  larger scale than along the main propagation direction. This
  Boussinesq system is proved to be consistent with the Navier-Stokes
  equations. This system is only an intermediate result that enables
  us to derive the Kadomtsev-Petviashvili (KP) equation which is a 2D
  generalization of the KdV equation. In addition, we get the 1D KdV
  equation, and lastly the Boussinesq equation. All these equations
  are derived for non-vanishing initial conditions.
\end{abstract}


\noindent \underline{Subject Class:} 74J15, 35Q30, 76M45, 35Q35\\
\underline{Keywords:} water waves, shallow water, Boussinesq system, viscosity, KdV equation, surface tension, KP\\

\section{Introduction}

\subsection{Motivation}

Understanding the evolution of water waves is a longstanding
problem. For instance, Stokes already discussed a precise wave at
least 130 years ago \cite{Stokes_1880} and his name has been given to
this wave. Such a wave was still discussed more than one century later
\cite{Hasimoto_Ono_72} !

One of the striking phenomena is that there exist surface waves that
travel almost without modification. Such a behavior is an impetus to
investigate the dynamics of these waves.

The first motivation is to understand how they may either
disperse or dissipate and yet not vanish. It was discovered that there
is a regime in which nonlinearity compensates dispersion (with no
dissipation).

Another motivation is more technical. Should a wave keep its shape for
a long time without vanishing, it could be used in telecommunication
without requiring repeaters. Such a new technology could drastically
reduce the cost of telecommunications.

Beyond mathematical and industrial considerations, mathematicians and
physicists model the water waves with a fluid flow in a channel of
(often) flat bottom, without meniscus (appart in
\cite{Mei_Liu_73}). Since the Euler or Navier-Stokes equations are
rather prone to faithfully model such a fluid, one could be satisfied
with either of these models. Yet, the Direct Numerical Simulation is
too expensive and asymptotic models are required. In such models, one
assumes a regime of dimensionless parameters' smallness and makes
expansions of the equations to derive a simplified model with less
fields and less dimensions.

\subsection{Literature}

The techniques are very different depending on whether one chooses
Euler or Navier-Stokes as a model for the fluid.

On the one hand, for an inviscid flow (Euler), one assumes very often
irrotationality. Then one may use a potential function and a Dirichlet to
Neumann operator in the Zakharov-Craig-Sulem formulation. Numerous
articles use this formulation and a review of the mathematical proofs
has been published \cite{Lannes_2013}. Some authors use a
velocity-pressure formulation, but they still assume the
irrotationnality assumption \cite{Iguchi_06}. The derivation of
Kadomtsev-Petviashvili (KP) assumes a predominant propagation
direction along $x$ and a weakly transverse propagation along $y$. It is
initially done for an inviscid fluid in
\cite{Kadomtsev–Petviashvili_70} and the justification of this
approximation is done in \cite{Lannes_2013} (subsection 7.2). In
\cite{Lannes_02}, D. Lannes proves that the sum of one wave
propagating to the right and one to the left, both obeying a KP
equation converges (when $\varepsilon$ tends to zero) on
$[0,T_0/\varepsilon]$ to a function consistent with the Boussinesq
system, meaning that it is a solution, should it be only up to
$O(\varepsilon^2)$.

Various difficulties must be considered with these surface wave
models. Firstly two scales of the transverse velocity are choosen in
the literature. Either one assumes the transverse velocity to be
$O(1)$ or $O(\sqrt{\varepsilon})$. This problem was raised in Remark 3
of \cite{Lannes_Saut_06} and here in Remark
\ref{rem_Lannes_Johnson}. It leads to different Boussinesq systems.

In \cite{Ming_Zhang_Zhang_12_a}, the authors prove that a
dimensionless water wave system in an infinite strip under the
influence of gravity and surface tension has a unique solution on
$[0,T/\varepsilon]$. More precisely, if the initial solution is
sufficiently regular and small ($O(\sqrt{\varepsilon})$), there exists
a solution, on this time interval $[0,T/\varepsilon]$, that will
remain small ($O(\sqrt{\varepsilon})$). In
\cite{Ming_Zhang_Zhang_12_b}, the same authors prove that on the same
time interval, these solutions can be accurately approximated by sums
of solutions of two decoupled KP equations. It enhances the results of
\cite{Lannes_Saut_06} by taking surface tension and a variable bottom
into consideration.

The Boussinesq {\em equation} is a generalization of the KdV equation
for waves moving both to the right and to the left. From the
Boussinesq system, instead of changing to a frame moving to the right
as is done for KdV, we remain in the same coordinate system and get
higher order equation (perturbed wave equation), as is done in
\cite{Johnson_97} (p. 216-219).\\

On the other hand, when one assumes the fluid to be viscous, one may
not use the potential function. Moreover the number of boundary
conditions changes then. Viscosity has been considered in water waves
since 1895 by Boussinesq \cite{Boussinesq_1895_1}. The dynamic of viscous
water waves on finite depth was investigated more recently by Kakutani
and Matsuuchi \cite{KM_75}. They derived the viscous KdV equation from
vanishing initial conditions but made the error of using a Fourier
transform in time while the problem is of Cauchy type.

Liu and Orfila \cite{Liu_Orfila_04}, and the coauthors of P.L.-F. Liu
in subsequent articles, investigated the viscous Boussinessq system
with vanishing initial condition, and validated their KdV equation by
some experiments. For instance they experimentally proved the reverse
flow in the boundary layer, which is predicted by theory.

Although it was not done in Sobolev spaces, the derivation
of the viscous Boussinesq system (1D and isotropic 2D) and 1D KdV
equation with a non-vanishing initial condition was done in
\cite{LeMeur_14} without surface tension. It was partially the goal of
\cite{LeMeur_14} not to rely on the irrotationnality assumption so as
to derive the Boussinesq system and KdV equation for a viscous
fluid. This could be achieved thanks to the fact that one of the
equations contains $u_z=O(\varepsilon)$ which is the zeroth order of
the irrotationnality assumption. In \cite{LeMeur_14}, a more detailed
bibliography was given and the reason why all the preceding articles
did not derive the same KdV equation is explained.

In this article, we intend to cross the results for a viscous fluid
with some results concerning inviscid fluids (KP, Boussinesq
equation, surface tension).\\

The outline of this article is as follows. In section \ref{sec.2}, we
present the equations and state the viscous Boussinesq system in 2D
with surface tension and a predominant propagation direction. This
result is needed to get the new results on viscous 1D KdV equations in
section \ref{sec.3}, the viscous KP equation (2D generalization of KdV
with a weak transverse propagation) derived in section \ref{sec.4},
and the viscous Boussinesq equation derived in section
\ref{sec.5}. Finally, we justify that our Boussinesq system is
consistent with the Navier-Stokes equations in Section \ref{sec.6}. All
these derivations are done for a viscous fluid with surface tension
and non-vanishing initial conditions.

\section{2D geometry with viscosity and surface tension}

\label{sec.2}

Below, we define the geometry, write the equations, and then make
these equations dimensionless (subsection \ref{subsec.2.1}). These
equations model a propagation predominantly in the $x$-direction. This
will enable us to make a linear analysis with the dispersion relation
and an asymptotic of the phase velocity in subsection
\ref{subsec.2.2}. A discussion of the relevance of viscosity and
surface tension will be given in subsection \ref{subsec.2.3}. We will
then use the results of \cite{LeMeur_14}, where the 1D viscous
isotropic Boussinesq system and KdV equation are derived, so as to
state the weakly transverse 2D viscous Boussinesq system (subsection
\ref{subsec.2.4}). Our present derivation includes surface tension,
viscosity and a propagation mainly along $x$ and weakly along
$y$. Since the already published results are very close, most proofs
are omitted.

\subsection{Navier-Stokes equations}

\label{subsec.2.1}

Let $\tilde{\vite}=(\tilde{u},\tilde{v},\tilde{w})$ be the velocity of a
fluid in a 3-D domain $\tilde{\Omega}= \{ (\tilde{x},\tilde{y},\tilde{z}) \; /
\; (\tilde{x},\tilde{y}) \in \mathbb{R}^2, \; \tilde{z} \in
(-h,\tilde{\eta}(\tilde{x},\tilde{y},\tilde{t}))\}$. So we assume the bottom is
flat and the free surface is characterized by $\tilde{z} =
\tilde{\eta}(\tilde{x},\tilde{y},\tilde{t})$ with
$\tilde{\eta}(\tilde{x},\tilde{y},\tilde{t})> -h$ (the bottom does not get dry).
Let $\tilde{p}$ be the pressure and
$\tilde{\mathbf{D}}[\tilde{\vite}]$ the symmetric part of the velocity
gradient. The dimensionless domain is illustrated in Fig. \ref{fig1}.
\begin{figure}[htbp]
\begin{center}
\includegraphics[width=7cm]{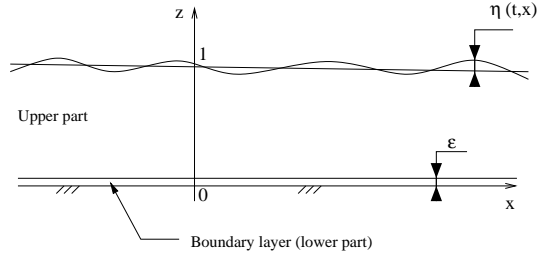}
\end{center}
\caption{The dimensionless domain}
\label{fig1}
\end{figure}
We also denote $\rho$ the density of the fluid, $\nu$ the viscosity of
the fluid, $g$ the gravitational acceleration, $\mathbf{k}$ the unit
vertical vector, $\mathbf{n}$ the outward unit normal to the upper
frontier of $\tilde{\Omega}$, $\sigma$ the surface tension
coefficient, and $\tilde{p}_{atm}$ the atmospheric pressure. The original
system reads:
\begin{equation}
\label{Ga.1}
\left\{
\begin{array}{ll}
\rho \left( \DP{\tilde{\vite}}{\tilde{t}} + \tilde{\vite} . \tilde{\nabla} \tilde{\vite} \right) -\nu \tilde{\Delta} \tilde{\vite} +\tilde{\mathbf{\nabla}} \tilde{p} = -\rho g \mathbf{k} & \mbox{ in } \tilde{\Omega}, \\
\widetilde{\mbox{div }} \tilde{\vite} = 0 & \mbox{ in } \tilde{\Omega}, \\
\left( -\tilde{p} \mathbf{I} +2 \nu \mathbf{\tilde{D}}[\tilde{\vite}] \right) . \mathbf{n}=-\tilde{p}_{atm} \mathbf{n} +\Frac{\sigma}{R} \mathbf{n}& \mbox{ on } \tilde{z} = \tilde{\eta}(\tilde{x},\tilde{t}), \\
\tilde{\eta}_{\tilde{t}} +\tilde{u} \tilde{\eta}_{\tilde{x}} +\tilde{v} \tilde{\eta}_{\tilde{y}} -\tilde{w} = 0 & \mbox{ on } \tilde{z} = \tilde{\eta}(\tilde{x},\tilde{t}), \\
\tilde{\vite} = 0 & \mbox{ on } \tilde{z} = -h,
\end{array}
\right.
\end{equation}
where we write the second order tensors and the vectors with bold
letters. The differentiated functions are denoted either $\partial f
/ \partial x$, $f_x$, or $\partial_x f$. The surface tension term
includes $R$ given by geometrical computations:
\[
\Frac{1}{R}= \Frac{(1+\tilde{\eta}_{\tilde{y}}^2)\tilde{\eta}_{\tilde{x}\tilde{x}}+(1+\tilde{\eta}_{\tilde{x}}^2)\tilde{\eta}_{\tilde{y}\tilde{y}}-2\tilde{\eta}_{\tilde{x}} \tilde{\eta}_{\tilde{y}} \tilde{\eta}_{\tilde{x}\tilde{y}}}{(1+\tilde{\eta}_{\tilde{x}}^2+\tilde{\eta}_{\tilde{y}}^2)^{3/2}}.
\]
Of course, we need to add an initial condition and conditions
at infinity.

So as to get dimensionless fields and variables, we need to choose a
characteristic horizontal length $l$ which is the wavelength along the
propagation direction, a characteristic vertical length $h$ which is
the water's height, and the amplitude $A$ of the propagating
perturbation. Moreover, let $U, V, W, P$ be the characteristic
horizontal velocity along $x$, horizontal velocity along $y$, vertical
velocity and pressure respectively. Since the zeroth order equation is
a wave equation of dispersion relation $\omega=c_0
k_x(1+k_y^2/2/k_x^2)+o(k_y^2)$, one may infer that a good scaling for
the $y$ direction is if the extra term is of the same order of
magnitude as $\varepsilon$. So $k_y=O(\varepsilon^{1/2})$ and the
interesting scale for $\tilde{y}$ is $O(\varepsilon^{-1/2})$ more than
for the $x$ direction. We may then define:
\[
\begin{array}{c}
  c_0 = \sqrt{gh}, \; \varepsilon = \Frac{A}{h}, \; \beta =\Frac{h^2}{l^2}, \; U =\varepsilon c_0, V= \varepsilon c_0, \\
  W=\sqrt{\varepsilon} c_0, \; P=\rho g A, \; \mbox{Re} =\Frac{\rho c_0 h}{\nu},\Bo=\Frac{\rho g l^2}{\sigma},
\end{array}
\]
where $c_0$ is the phase velocity. Then, one may make the fields
dimensionless and unscaled:
\[
\tilde{u} = U u, \; \tilde{v} = V v, \;\tilde{w}= W w,\; \tilde{p} = \tilde{p}_{atm} -\rho g \tilde{z} + P p, \; \tilde{\eta} = A \eta,
\]
and the variables too:
\begin{equation}
\label{Ga.4}
\tilde{x} = l x, \; \tilde{y} = \varepsilon^{-1/2} l y, \;\tilde{z}= h(z-1), \; \tilde{t}= t \, l/c_0.
\end{equation}
We also make the Boussinesq approximation ($\beta$ and $\varepsilon$
of the same order of magnitude) and even $\beta = \varepsilon$. Notice
that we take the scale used by \cite{Lannes_Saut_06} ($V=\varepsilon
c_0$) and not the one taken by \cite{Johnson_97} ($V=\varepsilon^{3/2}
c_0$). The difference will be highlighted in Remark
\ref{rem_Lannes_Johnson}.

With these definitions, the new system with
the new fields, variables and outward unit normal, still denoted
$\mathbf{n}$, writes in the new domain $\Omega_t=\{(x,y,z), (x,y) \in
\mathbb{R}^2, \; z\in (0, 1+\varepsilon \eta(x,y,t)) \}$ (see Figure
\ref{fig1}):
\begin{equation}
\label{Ga.5}
\left\{
\begin{array}{ll}
u_t+\varepsilon u u_x +\varepsilon^{3/2} v u_y + w u_z -\Frac{\sqrt{\varepsilon}}{\mbox{Re}}\left( u_{xx} +\varepsilon u_{yy}+\Frac{u_{zz}}{\varepsilon} \right) +p_x =0 & \mbox{ in } \Omega_t, \\
v_t+\varepsilon u v_x +\varepsilon^{3/2} v v_y + w v_z -\Frac{\sqrt{\varepsilon}}{\mbox{Re}} \left( v_{xx} +\varepsilon v_{yy}+\Frac{v_{zz}}{\varepsilon} \right) +\Frac{p_y}{\sqrt{\varepsilon}} =0 & \mbox{ in } \Omega_t, \\
w_t+\varepsilon u w_x +\varepsilon^{3/2} v w_y + w w_z -\Frac{\sqrt{\varepsilon}}{\mbox{Re}}(w_{xx} +\varepsilon w_{yy}+\Frac{w_{zz}}{\varepsilon}) +p_z =0 & \mbox{ in } \Omega_t,  \\
u_x +\sqrt{\varepsilon} v_y +w_z/\varepsilon =0& \mbox{ in } \Omega_t, \\
(\eta-p) \mathbf{n}  + \frac{1}{\mbox{Re}} \left( \begin{array}{ccc}
2\sqb u_x & (\varepsilon u_y+\sqb v_x) & u_z+ w_x  \\
(\varepsilon u_y+\sqb v_x) & 2 \varepsilon v_y & (v_z+\sqrt{\varepsilon} w_y)\\
u_z+ w_x &  (v_z+\sqrt{\varepsilon} w_y) & 2 \Frac{w_z}{\sqrt{\varepsilon}} \end{array} \right) . \mathbf{n}  & \\[7mm]
\hspace*{2cm}  = \Frac{1}{\Bo} \Frac{(1+\varepsilon^4 \eta_y^2)\eta_{xx}+\varepsilon (1+\varepsilon^3 \eta_x^2)\eta_{yy}-2\varepsilon^4 \eta_x \eta_y \eta_{xy}}{(1+\varepsilon^3 \eta_x^2+\varepsilon^4 \eta_y^2)^{3/2}} . \mathbf{n}   & \mbox{ on } z=1+\varepsilon \eta,\\
\eta_t+\varepsilon u \eta_x+\varepsilon^{3/2} v \eta_y- w/\varepsilon = 0 & \mbox{ on } z=1+\varepsilon \eta, \\
\vite=0 & \mbox{ on } z=0.
\end{array}
\right.
\end{equation}
Simple computations give the outward non-unit normal
$\mathbf{n}=(-\varepsilon \eta_x, - \varepsilon \eta_y, 1)$. 

\subsection{Linear theory}

\label{subsec.2.2}

The whole subsection below is a straightforward modification of the
case of a viscous fluid without surface tension, fully studied in
\cite{LeMeur_14}.

\subsubsection{Dispersion relation}
We are looking for small fields. So we linearize the system
(\ref{Ga.5}). The dynamic condition at the free boundary can then be stated
more explicitely. We get:
\begin{equation}
\label{Ga.6}
\left\{
\begin{array}{ll}
u_t -\frac{\sqrt{\varepsilon}}{\mbox{Re}}( u_{xx} +\varepsilon u_{yy} + u_{zz} / \varepsilon ) +p_x =0 & \mbox{ in } \mathbb{R}^2 \times [0,1], \\
v_t -\Frac{\sqrt{\varepsilon}}{\mbox{Re}}( v_{xx} +\varepsilon v_{yy} + v_{zz} / \varepsilon ) +\sqrt{\varepsilon} p_y =0 & \mbox{ in } \mathbb{R}^2 \times [0,1], \\
w_t -\Frac{\sqrt{\varepsilon}}{\mbox{Re}}( w_{xx} +\varepsilon w_{yy} + w_{zz} / \varepsilon ) +p_z =0 & \mbox{ in } \mathbb{R}^2 \times [0,1], \\
u_x +\sqrt{\varepsilon} v_y+w_z/\varepsilon =0 & \mbox{ in } \mathbb{R}^2 \times [0,1], \\
(u_z+w_x)/\RE = 0 & \mbox{ on } z=1, \\
(v_z+\sqrt{\varepsilon} w_y)/\RE = 0 & \mbox{ on } z=1, \\
(\eta -p) +\Frac{2  w_z}{\mbox{Re}\sqrt{\varepsilon}}=\Frac{1}{\Bo}(\eta_{xx}+\varepsilon \eta_{yy}) & \mbox{ on } z=1, \\
\eta_t- w/\varepsilon = 0 & \mbox{ on } z=1,\\
\vite=0 & \mbox{ on } z=0.
\end{array}
\right.
\end{equation}
In order to eliminate the pressure, we differentiate (\ref{Ga.6})$_1$
with respect to $z$ and (\ref{Ga.6})$_3$ with respect to $x$ and
compute their difference so as to simplify $p_{xz}$. Symmetrically, we
differentiate (\ref{Ga.6})$_2$ with respect to $z$ and
(\ref{Ga.6})$_3$ with respect to $y$ and compute their difference
(with a $\sqrt{\varepsilon}$ coefficient). Pressure is eliminated in
two equations on three:
\[
\begin{array}{l}
u_{zt}- w_{xt}-\frac{\sqb}{\RE}\left( (u_{xxz}-w_{xxx})+ \varepsilon (u_{yyz}-w_{yyx})+(u_{zzz}-w_{xzz})/\varepsilon \right) =0,\\
v_{zt}-\sqrt{\varepsilon} w_{yt}-\frac{\sqb}{\RE}\left( (v_{xxz}-\sqb w_{xxy})+ \varepsilon (v_{yyz}-\sqb w_{yyy})+(v_{zzz}- \sqb w_{yzz})/ \varepsilon \right) =0.
\end{array}
\]
So as to eliminate $u$ and $v$ thanks to the incompressibility
(\ref{Ga.6})$_4$ from the previous system, we differentiate the first
equation with respect to $x$ and the second with respect to $y$. After
some simplifications, a good combination of the two gives (thanks to
(\ref{Ga.6})$_4$):
\begin{equation}
\label{Ga.7}
(\partial_z^2 +\varepsilon \partial_x^2+ \varepsilon^2 \partial_y^2)(-\RE \sqb \partial_t +(\partial_z^2 +\varepsilon \partial_x^2+ \varepsilon^2 \partial_y^2))w=0.
\end{equation}
Let $w$ be of the form $\mathcal{A}(z) \exp{i(k_x x+k_y y -\omega t)}$
with non-negative $k_x,k_y$ and a (complex) pulsation $\omega$. We can
define a complex parameter with non-negative real part, similar to the
one used by \cite{KM_75}:
\begin{equation}
\label{Ga.7.5}
\mu^2=\varepsilon k^2 - i \omega \RE \sqb ,
\end{equation}
where $k^2=k_x^2+\varepsilon k_y^2$. Thanks to this notation, the solutions
of (\ref{Ga.7}) are such that
\begin{multline}
\label{Ga.8}
\mathcal{A}(z)=  C_1 \cosh{\sqb k(z-1)} +C_2 \sinh{\sqb k(z-1)} \\
    +C_3 \cosh{\mu (z-1)} +C_4 \sinh{\mu(z-1)}.
\end{multline}
The surface tension did not appear yet because it is only in the
boundary conditions.

Up to now we have eliminated $u, v$ and $p$ only in the volumic
equations. We still have to use the boundary conditions of
(\ref{Ga.6}) to find the conditions on the remaining field $w$.

Easy computations on the boundary conditions of (\ref{Ga.6}), similar
to the ones done in \cite{LeMeur_14}, give:
\begin{equation}
\label{Ga.9}
\begin{array}{l}
w_z(0)=0, \\
w(0)=0, \\
w_{zz}(1)-\varepsilon w_{xx}(1)- \varepsilon^2 w_{yy}(1)=0, \\
w_{xx}+\varepsilon w_{yy}(1) -w_{ztt}(1)+\frac{3 \sqb}{\RE} (w_{xxzt}+\varepsilon w_{yyzt})+\frac{1}{\RE \sqb} w_{zzzt}(1) =\frac{1}{\Bo}(\partial_x^2+\varepsilon \partial_y^2)^2w .
\end{array}
\end{equation}
The solutions (\ref{Ga.8}) satisfies a homogeneous linear system in
the constants $C_1, C_2, C_3, C_4$. Its matrix is:

\begin{equation}
\label{Ga.9.5}
\left(
\begin{array}{cccc}
\sqb k\sinh{(\sqb k)} & -\sqb k \cosh{(\sqb k)} & \mu \sinh{\mu} & -\mu \cosh{\mu} \\
\cosh{(\sqb k)} & -\sinh{(\sqb k)} & \cosh{\mu} & -\sinh{\mu} \\
2k^2 \varepsilon & 0 & \mu^2+\varepsilon k^2 & 0 \\
-k^2-\frac{k^4}{\Bo} & \sqb \omega^2 k +\frac{2 i \varepsilon \omega k^3}{\RE} & -k^2 -\frac{k^4}{\Bo} & \frac{2 \mu \sqb i \omega k^2}{\RE}
\end{array}
\right),
\end{equation}
where $k^2=k_x^2+\varepsilon k_y^2$. It suffices to compute its
determinant to get the dispersion relation:
\begin{multline}
\label{Ga.10}
4\varepsilon k^2 \mu (\varepsilon k^2+\mu^2) +4\mu k^3 \varepsilon^{3/2}(\mu \sinh{(k\sqb)} \sinh{\mu}-k \sqb \cosh{(k\sqb)} \cosh{\mu}) \\
-(\varepsilon k^2 +\mu^2)^2(\mu \cosh{(k\sqb)} \cosh{\mu}-k \sqb \sinh{(k \sqb)} \sinh{\mu}) \\
-(k+k^3/\Bo) \sqb \RE^2(\mu \sinh{(k\sqb)}  \cosh{\mu}-k\sqb \cosh{(k\sqb)} \sinh{\mu})=0.
\end{multline}
This relation is identical to the one of \cite{LeMeur_14} appart from
the surface tension term and the new definition of $k$ for the 3D
geometry ($k^2=k_x^2+\varepsilon k_y^2$). Our process of non-dimensionnalizing
makes a difference between $x$ and $z$. This is why \cite{KM_75} have
$k$ terms, and we have $k \sqb$ terms instead.

\subsubsection{Asymptotic of the phase velocity (very large Re)}

In this subsection, we state the following Proposition concerning the
phase velocity:
\begin{proposition}
\label{prop1}
Under the assumptions
\begin{align*}
& k \sqb \, \RE \; c  \rightarrow + \infty, \\
& k =  \; O(1), \\
&\varepsilon \rightarrow  \;0, \\
&\RE \rightarrow  + \infty, \\
& c =  \;O(1)\; \;  (\mbox{and } c \mbox{ bounded away from } 0),\\
& \Bo \mbox{ bounded away from } 0,
\end{align*}
if there exists a complex phase velocity $c=\omega/k$ solution of (\ref{Ga.10}),
then it is such that:
\begin{equation}
\label{tl.2}
c =\sqrt{\left(1+\Frac{k^2}{\Bo} \right) \Frac{\tanh{(k \sqb)}}{k \sqb}}-\Frac{e^{i \pi
/4}\, \RE^{-1/2}(k \sqb)^{1/4}}{2\tanh^{3/4}{(k\sqb)}} + o(\varepsilon^{-1/4} \RE^{-1/2}).
\end{equation}
Moreover, the decay rate in our finite-depth geometry does not depend
on the surface tension but only on the viscosity. It is:
\begin{equation}
\label{tl.2.5}
\mbox{Im}(\omega)=\mbox{Im}(kc)=\frac{-1}{2\sqrt{2}} \frac{k^{5/4}\varepsilon^{1/8}}{\sqrt{\RE}\tanh^{3/4}{(k\sqb)}}+o(\varepsilon^{-1/4}\RE^{-1/2}).
\end{equation}
\end{proposition}
Notice that our geometry is not infinite. So our formula may be different
from Boussinesq's and Lamb's.

Appart from the surface tension terms and the change of notation of
$k$, the proof is not different from the one of Proposition 1 in
\cite{LeMeur_14}.

To what extent is surface tension relevant ?

\subsection{Relevance of surface tension and viscosity}
\label{subsec.2.3}

As was explained in the previous subsection, surface tension yields
only a real term in the phase velocity. So it does not give
dissipation, but only dispersion. On the contrary, viscosity
influences both the real and the imaginary part of the phase velocity.

We choose two different fluids (water and mercury) in order to discuss
the relevance of surface tension and viscosity. Their physical
parameters are listed in Table \ref{table_1}.

\begin{table}[htb]
\centering
\begin{tabular}{|c|c|c|c|c|c|c|c|}
\hline
 & $\sigma (Nm^{-1})$ & $\rho ({\rm kg}\,m^{-3})$ & $\nu ({\rm Pa}\, s)$ & 1/Re & 1/Bo \\
\hline
water $20^o C$ &  $7.3\,10^{-2}$ & $10^{3}$ & $10^{-3}$ & $3.2\,10^{-7}/h^{3/2}$ & $7.4\,10^{-6}/l^2$ \\
mercury $20^o C$ &  $4.4\,10^{-1}$ & $1.3\,10^{4}$ & $1.56\,10^{-3}$ & $3.8\,10^{-8}/h^{3/2}$ & $3.5\,10^{-6}/l^2$ \\
\hline
\end{tabular}
\caption{Physical parameters for two fluids}
\label{table_1}
\end{table}
Given $\sigma, \rho, \nu$, the Re and Bo number are evaluated in Table
\ref{table_1}. Depending on $h$ or $l$, the Re or Bo dimensionless
parameters are greater or less than 1. The critical values for each
of them are given in Table \ref{table_2}.
\begin{table}[htb]
\centering
\begin{tabular}{|c|c|c|}
\hline
 & $h_{crit-\RE}$& $l_{crit-\Bo}$  \\
\hline
water $20^o C$ & $4.6\,10^{-5}$ & $2.7\,10^{-3}$\\
mercury $20^o C$ & $1.1\,10^{-5}$ & $1.9\,10^{-3}$\\
\hline
\end{tabular}
\caption{Critical values of length (in $m$) for which 1/Re or 1/Bo is 1.}
\label{table_2}
\end{table}
It appears that viscosity is not so relevant as surface
tension. Furthermore, so as to compare $1/$Bo and $\varepsilon$, and
under the Boussinesq approximation, one may compute the ratio
\[
\frac{1}{\Bo}/ \varepsilon=\frac{\sigma}{\rho g h^2}.
\]
The critical value of $h$, for which this ratio is 1, is
$h_{app}=\sqrt{\sigma/(\rho g)}$. This other critical value
is either $2.7\, mm$ for water or $1.9\, mm$ for mercury. In other words
if $h>h_{app}$, then the surface tension term is irrelevant
and even less important than gravity measured by $\varepsilon$ in the
equation. In most current experiments, $h$ is more than $0.1m$ and so,
surface tension should (globally) not be taken into account. As
D. Lannes says: ``for coastal waves of characteristic length $L_x=10m$,
capillary effects represent only $0.0003\%$ of the gravity effects''
(\cite{Lannes_2013} section 9.1.2).

Surface tension is irrelevant for flows with $h_0$ more than $1 \,
cm$. Yet, even if the initial flow satisfies this, it may or may not
remain true. For instance, in case of flows with ripples created by
the wind or in case of wave breaking, very short wavelengths
appear. See section 9.1.2 of \cite{Lannes_2013} for some references.

In the case of a flow in a shallow channel (some $mm$ height), and for
mercury, a depression solitary wave was predicted (as early as the
initial article of Korteweg and de Vries \cite{KdV}) and observed in
\cite{Falcon_02}. In this experiment, $A=0.064\, mm$, $h=2.12\,
mm$. As a consequence, Re=$2547$ and viscosity should be taken into
account.

In the following, we assume $1/\Bo$ between $O(\varepsilon)$ ($l \sim
1\, mm$) and $O(\varepsilon^2)$ ($l \sim 1\,cm$).

\subsection{Viscous Boussinesq System}
\label{subsec.2.4}

Let us first discuss the regime binding $\varepsilon, \RE, \Bo$.

The complex phase velocity (\ref{tl.2}) contains gravitational and
viscous terms that we want to compare so as to find the regime at
which their variations are of the same order of magnitude. The first
term is the gravitational term ($\sqrt{(1+k^2/\Bo)\tanh{(k \sqb)} /(k
  \sqb)}$) which may be expanded when $\varepsilon$ tends to zero:
$1-k^2 \varepsilon( 1-3/(\varepsilon \Bo))/6 +
O(\varepsilon^2)+O(\Bo^{-2})$. The second is purely viscous and can be
expanded: $-\sqrt{2}(1+i)(4 \sqrt{k})^{-1}(\RE \, \sqb)^{-1/2}+o(\RE
\, \sqb)^{-1/2})$. So if $1/\Bo = O(\varepsilon)$ or less but
$1/(\varepsilon \Bo ) \not \sim 1/3$, the variations of $c$ on the
gravitational and on the viscous effects are of the same order of
magnitude when $\varepsilon(1-3/(\varepsilon \Bo))$ and $(\RE \,
\sqb)^{-1/2}$ are of the same order. In this regime of very large Re,
studied hereafter, the dependence of $\RE$ on $\varepsilon$ is such
that (to simplify, we assume here $1/\Bo =o(\varepsilon)$):
\begin{equation}
\label{tl.7}
\RE \sim \varepsilon^{-5/2}.
\end{equation}
This would be wrong if $1/(\varepsilon \Bo) \sim 1/3$ as is well-known for
non-viscous fluids. In that case, we would need to go further in the
expansion. But since it was assumed above that $1/ \Bo$ is between
$O(\varepsilon^2)$ and $O(\varepsilon)$, we are driven to choose the
regime (\ref{tl.7}) and exclude $1/(\varepsilon \Bo)\sim 1/3$. 

Our main purpose here is to state an asymptotic system of reduced size
from the global Navier-Stokes equations in the whole {\em moving}
domain with surface tension and a predominant propagation direction.
\begin{proposition}
\label{Prop.1}
Let $\mathbf{x}=(x,y)\in \mathbb{R}^2$. Let $\eta(\mathbf{x},t)$ be
the free boundary's height. Let
$(u^{b,0}(\mathbf{x},\gamma),v^{b,0}(\mathbf{x},\gamma))$ for $\gamma
\in (0,+\infty)$
(resp. $(u^{u,0}(\mathbf{x},z),v^{u,0}(\mathbf{x},z))$ for $z\in
(0,1+\varepsilon \eta(\mathbf{x},t))$) be the initial horizontal
velocity in the boundary layer (resp. in the upper part of the
domain). Let
$f(\mathbf{x},\gamma):=u^{b,0}(\mathbf{x},\gamma)-u^{u,0}(\mathbf{x},z=0)$
and
$g(\mathbf{x},\gamma):=v^{b,0}(\mathbf{x},\gamma)-v^{u,0}(\mathbf{x},z=0)$. If
for any $\mathbf{x}\in \mathbb{R}^2$, $\gamma \mapsto f$ and $\gamma
\mapsto g$ are uniformly continuous in $\gamma$ and $f_x, g_y \in
L^1_{\gamma}(\mathbb{R}^+)$, then the solution of the Navier-Stokes
equation with this given initial condition satisfies the weakly
transverse Boussinesq system:
\begin{equation}
\label{Boussinesq_1}
\begin{array}{l}
u_t+\eta_x -\Frac{\eta_{xxx}}{\Bo}+\varepsilon u u_x+\varepsilon^{3/2}v u_y-\varepsilon \eta_{xtt}\Frac{(z^2-1)}{2} =  O(\varepsilon^2)+O(\varepsilon/\Bo),\\
v_t+\sqb \eta_y -\Frac{\sqb \eta_{xxy}}{\Bo} +\varepsilon u v_x+\varepsilon^{3/2}v v_y-\varepsilon \eta_{ytt}\Frac{(z^2-1)}{2} =  O(\varepsilon^2)+O(\varepsilon/\Bo),\\
\eta_t+u_x(\mathbf{x},z,t)-\Frac{\varepsilon}{2} \eta_{xxt} (z^2-\frac{1}{3})+ \varepsilon (u \eta)_x +\sqb v_y+\varepsilon^{\frac{3}{2}}(v \eta)_y-\Frac{\varepsilon}{\sqrt{\pi R}} (u_x+\sqb v_y) \ast \Frac{1}{\sqrt{t}}   \\
\hspace*{7mm}+\Frac{2\varepsilon}{\sqrt{\pi}}\Int_{\gamma''=0}^{+\infty} \mbox{\rm div } \left(\begin{array}{c}
(u^{b,0}(\mathbf{x},\gamma'') -u^{u,0}(\mathbf{x},z=0))\\
\sqb (v^{b,0}(\mathbf{x},\gamma'') -v^{u,0}(\mathbf{x},z=0))
\end{array}\right)
 \Int_{\gamma'=0}^{\sqrt{\frac{\Rb}{4t}}\gamma''} e^{ -\gamma'^2}{\rm d}\gamma'{\rm d}\gamma'' =  o(\varepsilon),
\end{array}
\end{equation}
where the convolution, denoted with $\ast$, is in time, the parameters
$\varepsilon, $ and $R$ have been defined and $z \in (0,1+\varepsilon
\eta(\mathbf{x},t))$.
\end{proposition}
With the change of notation $k^2=k_x^2+k_y^2$, the whole proof is the
very same as in \cite{LeMeur_14} and is omitted.

Had we proceeded in the same way as \cite{KM_75}, we would have
distinguished two subdomains: the upper part ($z > \varepsilon$) where
viscosity can be neglected, and the lower part ($0 < z < \varepsilon$)
which is a boundary layer at the bottom and where viscosity must be
taken into account. The resolution in each part would have enabled to
match the solution on the common boundary. Instead, we assume
overlapping domains.

\begin{remark}
\label{rem_Lannes_Johnson}
As is stressed at the non-dimensionnalizing step, there are two
choices for the scale of the transverse horizontal velocity. Either
one takes $V=\varepsilon^{3/2} c_0$ (see Johnson's book
\cite{Johnson_97} for instance or many others), or one uses
$V=\varepsilon c_0$ (see the article of Lannes-Saut
\cite{Lannes_Saut_06} for instance). Firstly one must notice that
since this field is not eliminated (as $w$ is), its scale is
relevant. Moreover the latter scaling is more general than the former
one because it does not assume $v$ to be $O(\sqrt{\varepsilon})$. The
main difference is that $v_{Lannes-Saut}=\sqb v_{Johnson}$, where one
implicitely assumes that the fields $v_{Lannes-Saut}$ and
$v_{Johnson}$ are $O(1)$. To be clear, Lannes-Saut's choice amounts to
assume $v=O(1)$ although it can be written
$\sqrt{\varepsilon} \partial_y \psi$ ! If we plug this
$v=\sqrt{\varepsilon} \partial_y \psi$ in (\ref{Boussinesq_1}), we get
the same transverse Boussinesq system as \cite{Johnson_97} (p. 210)
because $\varepsilon^{3/2} v_{Lannes-Saut}$ transforms into
$\varepsilon^2 v_{Johnson}$ and go to $O(\varepsilon^2)$. So the
non-dimensionnalizing process of Lannes-Saut is weaker than the one
done by many other authors.

Lannes and Saut \cite{Lannes_Saut_06} give one more argument by
exhibiting an (odd) solution to the linearized system (in their Remark
3) such that $v=\sqb \partial_y \psi$. It is of size $O(1)$ and not
$O(\sqrt{\varepsilon})$:
\[
u=0, \; v=f'(y-\sqrt{\varepsilon}t)-f'(y-\sqrt{\varepsilon}t), \eta=f'(y-\sqrt{\varepsilon}t)+f'(y-\sqrt{\varepsilon}t).
\]
We want to notice here that the potential representation of the
velocity might have induced the idea that if the potential is
physical, then the deduced field (the velocity) is such that
$(u,v)=(\partial_x \phi, \sqb \partial_y \phi)$ and so, that
$v=O(\sqb)$. Had we started from the velocity pressure formulation,
this assumption would not be natural. Yet, such an attempt of
explanation does not match reality. Johnson \cite{Johnson_97} uses the
velocity pression representation as we do and assumes
$v=O(\sqrt{\varepsilon})$, while Lannes and Saut \cite{Lannes_Saut_06}
use the potential representation and do not make this assumption !
\end{remark}

\begin{remark}
\label{remark_Euler_NS}
It was already noticed in \cite{LeMeur_14} that the double integral
term in (\ref{Boussinesq_1}) is new and surprising because of its
dependence on the initial condition. If we assume an initial flow of
Euler type (so that $u^{b,0}-u^{u,0}=0$), the double integral
vanishes. But is it physical ? In other words, does an initial
(inviscid) flow in the boundary layer (where Navier-Stokes applies)
establish (as a Navier-Stokes flow) fast or not ? Can we assume an
initial Euler flow and a Navier-Stokes evolution PDE without loss of
generality ?

We claim that the answer is negative for at least two reasons. On the
one hand, the characteristic time for the viscous effects to appear is
roughly $T_{NSE}=\rho h_0^2/\nu$ or $T_{NSE}=\rho l^2/\nu$. Then, its
ratio with the characteristic time of the inviscid gravity flow
($l/c_0$) is either Re $\sqrt{\varepsilon}=\varepsilon^{-2}$ or $\RE /
\sqrt{\varepsilon}=\varepsilon^{-3}$ respectively. Whatever the chosen
characteristic time, this ratio is large and the flow in the boundary
layer does not establish fast enough.

On the other hand, as is argued in \cite{LeMeur_14}, the value of this
integral for a typical exponential flow in the boundary layer can be
computed. It tends to zero only like $1/\sqrt{t}$ as is classical for
Navier-Stokes flows. As a consequence, one may assume that it goes
from zero (Euler) to exponential (boundary layer) within the time
($1/\sqrt{t}$) which is large with respect to the characteristic time
of the flow.
\end{remark}

\section{Viscous KdV with surface tension}
\label{sec.3}

Starting from (\ref{Boussinesq_1}), one may assume the fields do not
depend on $y$, and $v=0$. So the flow is purely one dimensional (one
direction). One may see that the zeroth order of this Boussinesq
system is the wave equation. Then one knows that there are two waves
propagating in each direction. If we look only for the waves that
propagate to the right, one may make a change of variables suggested by
the zeroth order equation:
\[
\xi = x-t, \; \tau= \varepsilon t.
\]
Every term can easily be converted in these new coordinates appart
from the convolution (one derivative and a half integration), the new
surface tension term and the double integral. The first was treated in
a very clean way in \cite{LeMeur_14}. The second is new, but very
simple. The third one was treated in \cite{LeMeur_14}, but its
treatment is improved below. After a change of variable, this 
term writes:
\begin{equation}
\label{KdV.1}
+\Frac{2\varepsilon}{\sqrt{\pi}}\Int_{\gamma''=0}^{+\infty} \left(u^{b,0}_{x}(\xi+\frac{\tau}{\varepsilon},\gamma'') -u^{u,0}_{x}(\xi+\frac{\tau}{\varepsilon},z=0) \right) \times \Int_{\gamma'=0}^{\sqrt{\frac{R\, \varepsilon}{4\tau}} \gamma''} e^{ -\gamma'^2}{\rm d}\gamma'{\rm d}\gamma''
\end{equation}
The article \cite{LeMeur_14} argued that if the initial horizontal
velocity is localized and $\tau$ not too small, then
$u^{b,0}(\xi+\tau/\varepsilon,\gamma'')-u^{u,0}(\xi+\tau/\varepsilon,z=0)$
will be negligible in comparison with $\varepsilon$. More precisely,
if $\xi \mapsto u^{b,0}(\xi,\gamma'')-u^{u,0}(\xi,z=0)$ tends to zero
when $\xi$ tends to $+\infty$ at least as $\xi^{-1}$. Then the term
(\ref{KdV.1}) is at least $O(\varepsilon^2)$.

One may add one more argument if $\tau$ is still not too small. The
upper bound of the inner integral contains a $\sqrt{\varepsilon}$. So
mainly large $\gamma''$ (larger than $1/\sqrt{\varepsilon}$) will be
relevant. But for those $\gamma''$, the
$u^{b,0}(\xi+\tau/\varepsilon,\gamma'')-u^{u,0}(\xi+\tau/\varepsilon,z=0)$
is small uniformly in $\xi+\tau / \varepsilon$ because of the matching
condition on $u^0$ between the boundary layer and the upper part. So
(\ref{KdV.1}) is small for two different reasons and we can state the
viscous KdV equation with surface tension in the following
Proposition.
\begin{proposition}
\label{prop_KdV}
If the initial flow is localized, the KdV change of variables applied
to the 1D version of the system (\ref{Boussinesq_1}) leads to
\begin{equation}
\label{kdv}
2\tilde{\eta}_{\tau}+3 \tilde{\eta}
\tilde{\eta}_{\xi}+\left( \frac{1}{3}-\frac{1}{\varepsilon \Bo} \right)\tilde{\eta}_{\xi \xi \xi}
-\frac{1}{\sqrt{\pi R}}\int_{\xi'=0}^{\tau/\varepsilon}\frac{\tilde{\eta}_{\xi}(\xi+\xi',\tau)}{\sqrt{\xi'}}{\rm d}\xi'=o(1),
\end{equation}
for not too small times $\tau$, where we set 
Re$=R\,  \varepsilon^{-5/2}$.
\end{proposition}

We do not replace the upper bound of the convolution ($\tau /
\varepsilon$) by $+\infty$ for two reasons. On the one hand we do not
know how fast this integral converges when $\tau / \varepsilon$ tends
to $+\infty$. On the other hand the $\tau$ term reminds us that this
integral on $\xi'$ is a mixture of time and space.

\section{Viscous KP equation with surface tension}
\label{sec.4}

Below, we derive the viscous KP equation, then discuss the zero-mass
(in $\xi$) constraint.

\subsection{Derivation}
We recall that we set Re$=R \, \varepsilon^{-5/2}$,
and $1/ \Bo=O(\varepsilon)$ but $1/(\varepsilon
\Bo)\not = 1/3+o(1)$.

\noindent With these assumptions, the change of variables
\[
\xi=x-t, \; y=y, \; \tau= \varepsilon t,
\]
and the assumption that the initial flow is localized, the system
(\ref{Boussinesq_1}) writes (since $\varepsilon/\Bo=O(\varepsilon^2)$):
\begin{equation}
\label{KP.1}
\begin{split}
-u_{\xi}(\xi,y,z,\tau)+\eta_{\xi}+\varepsilon (u_{\tau}+u u_{\xi}- \eta_{\xi\xi\xi}(z^2-1)/2 -\Frac{\eta_{\xi\xi\xi}    }{\varepsilon \Bo}+\sqb vu_y) & =  O(\varepsilon^2),\\
-v_{\xi}+\sqb \eta_y +\varepsilon (v_{\tau}+u v_{\xi}-\sqb \eta_{y \xi \xi}(z^2-1)/2 -\Frac{\eta_{\xi \xi y}}{\sqb \Bo} +\sqb v v_y )& = O(\varepsilon^2),\\
-\eta_{\xi}+u_{\xi}+\varepsilon (\eta_{\tau}+\Frac{1}{2} \eta_{\xi \xi \xi} (z^2-1/3)+ (u \eta)_{\xi} +\sqb (v \eta )_y )+\sqb v_y & \\
-\Frac{\varepsilon}{\sqrt{\pi R}} \Int_{\xi'=0}^{\tau/ \varepsilon} \Frac{u_{\xi}(\xi +\xi',y,z,\tau)+ \sqb v_{y}(\xi +\xi',y,z,\tau)}{\sqrt{\xi'}} {\rm d} \xi' & =  O(\varepsilon^2).
\end{split}
\end{equation}
In the above system, we did not write the double integral term from
(\ref{KdV.1}). Indeed, it was justified in the previous section that
this term can be neglected. Moreover, one might prove that the
transverse velocity is such that $v(z)=v(z')+O(\varepsilon)$ in a
similar way to Lemma 11 of \cite{LeMeur_14} (roughly, one
differentiate (\ref{Boussinesq_1})$_2$ with respect to $z$, then
integrate with respect to time $t$). So as to go further, since, at
the first order $u_{\xi}=\eta_{\xi}$, one may justify that
$u_{\tau}=\eta_{\tau}+o(1)$ as in the inviscid case, and
$u=\eta+o(1)$. One may then add the first and third equation:
\begin{align}
\nonumber 2\varepsilon \eta_{\tau}+ 3\varepsilon \eta \eta_{\xi} +\left(\frac{\varepsilon}{3} -\frac{1}{\Bo} \right) \eta_{\xi \xi \xi}+\varepsilon^{3/2}v u_y +\varepsilon^{3/2}(v \eta )_y+\sqb v_{y}& \\
\label{eq19.5}
-\Frac{\varepsilon}{\sqrt{\pi R}} \Int_{\xi'=0}^{\tau/ \varepsilon} \Frac{\eta_{\xi}(\xi +\xi',y,\tau)+\sqb v_y(\xi +\xi',y,z,\tau)}{\sqrt{\xi'}} {\rm d} \xi'& =O(\varepsilon^2).
\end{align}
In order to eliminate $v$, one must differentiate with respect to
$\xi$ to use the corresponding (\ref{KP.1})$_2$. Our choice of
non-dimensionalization forces us to manage extra terms, but we get the
classical KP equation with a viscous term:
\[
\begin{split}
\left( 2\eta_{\tau}+ 3\eta \eta_{\xi} +\left(\frac{1}{3} -\frac{1}{\varepsilon \Bo} \right) \eta_{\xi \xi \xi}\right)_{\xi}+\eta_{yy}-\Frac{1}{\sqrt{\pi R}} \Int_{\xi'=0}^{\tau/ \varepsilon} \Frac{\eta_{\xi \xi}(\xi +\xi',y,\tau)}{\sqrt{\xi'}} {\rm d} \xi'\\
=-\sqb \left(2v\eta_{\xi y}+v_y \eta_{\xi}+v_{\tau y} \right) + O(\varepsilon)+O(1/\Bo),
\end{split}
\]
where $v$ satisfies (\ref{KP.1})$_2$. This completes the proof of the
following Proposition.
\begin{proposition}
\label{prop_KP}
If the initial flow is localized, in the Boussinesq approximation and
if $1/\Bo =O(\varepsilon)$ but $1/\Bo \neq
\varepsilon/3+o(\varepsilon)$, the surface waves of a viscous fluid
propagating predominantly along $x$ and weakly along $y$ satisfy the
viscous KP equation:
\begin{equation}
\label{KP}
\left( 2\eta_{\tau}+ 3\eta \eta_{\xi} +\left(\frac{1}{3} -\frac{1}{\varepsilon \Bo} \right) \eta_{\xi \xi \xi}\right)_{\xi}+\eta_{yy}-\Frac{1}{\sqrt{\pi R}} \Int_{\xi'=0}^{\tau/ \varepsilon} \Frac{\eta_{\xi \xi}(\xi +\xi',y,\tau)}{\sqrt{\xi'}} {\rm d} \xi'= O(\sqrt{\varepsilon}),
\end{equation}
or the equivalent mixed system
\begin{equation}
\label{KP_system}
\begin{split}
2\eta_{\tau}+ 3\eta \eta_{\xi} +\left(\frac{1}{3} -\frac{1}{\varepsilon \Bo} \right) \eta_{\xi \xi \xi}-\Frac{1}{\sqrt{\pi R}} \Int_{\xi'=0}^{\tau/ \varepsilon} \Frac{\eta_{\xi}(\xi +\xi',y,\tau)}{\sqrt{\xi'}} {\rm d} \xi'+v_{y} & =O(\sqb),\\
-v_{\xi} +\eta_y & = O(\sqrt{\varepsilon}).
\end{split}
\end{equation}
\end{proposition}
%
%
One must keep in mind that our choice of scale triggers the
$\sqrt{\varepsilon}$ (even inside the convolution) in
(\ref{KP.1}). This choice is already discussed in Remark
\ref{rem_Lannes_Johnson}. Had we made the same choice of scale for $v$
as Johnson \cite{Johnson_97} and many others, we would have replaced
$v$ by $\sqrt{\varepsilon} v_J$ in (\ref{KP.1}). The equivalent
(\ref{eq19.5}) could be simplified by $\varepsilon$ and it would write
:
\[
2\eta_{\tau}+ 3\eta \eta_{\xi} +\left(\frac{1}{3} -\frac{1}{\varepsilon\Bo} \right) \eta_{\xi \xi \xi}+ v_{Jy} -\Frac{1}{\sqrt{\pi R}} \Int_{\xi'=0}^{\tau/ \varepsilon} \Frac{\eta_{\xi}(\xi +\xi',y,\tau)}{\sqrt{\xi'}} {\rm d} \xi' =O(\varepsilon),
\]
and $v_J$ would satisfy (\ref{KP.1})$_2$ modified:
\[
-v_{J\xi}+\eta_y +\varepsilon \left( v_{J\tau}+u v_{J\xi}- \eta_{y \xi \xi}(z^2-1)/2 -\Frac{\eta_{\xi \xi y}}{\varepsilon \Bo} \right) = O(\varepsilon^{3/2}).
\]
Oddly, this would lead to the same KP equation as (\ref{KP}) but up to
the order $O(\varepsilon)$ and not $O(\sqrt{\varepsilon})$. During the
proof, we exhibit the $O(\sqrt{\varepsilon})$ terms which would be
$O(\varepsilon)$ if $v=\sqrt{\varepsilon} v_J$ although they are not
kept in the Proposition. So there is no contradiction and the accuracy
of KP is tied to the property that $v=O(\sqrt{\varepsilon})$ or not.

\subsection{The zero-mass constraint}

As is very well discussed in \cite{MST_07}, the usual KP equation is
often written with an operator $\partial_{\xi}^{-1} \partial_y^2$. Yet
such an operator assumes the solution is differentiated (in $\xi$)
from a function that tends to $0$ when $\xi \rightarrow \pm
\infty$. It is proved in \cite{MST_07} that although it is not
obvious, this assumption is right for inviscid fluids. Their proof is
done in two steps.

In the first step, the linear KP equation is solved with inverse
Fourier transforms. For more general equations, but only dispersive,
like $u_t-Lu_{\xi}+\partial_{\xi}^{-1}\partial_{y\, y}u=0$ (the symbol
of $L$ is $\varepsilon \mid X \mid^{\alpha}$ for $\alpha >1/2$), the
fundamental solution writes
\[
G(t,\xi,y)=\mathcal{F}^{-1 }_{(X,y') \rightarrow (\xi,y)}\left[e^{it(\varepsilon X | X|^{\alpha}-y'^2/X)} \right],
\]
where $\varepsilon= \pm 1$ depends on the KPI or KPII equation (or on
the sign of $1/3-1/(\varepsilon \Bo)$). For the usual KPI and KPII,
$\alpha = 2$.

\noindent Thanks to the Lebesgue's Dominated convergence theorem and
inventive changes of variables, it is proved that the fundamental
solution is regular, and is differentiated from a more regular
function $A$ (which is $\mathcal{C}(\mathbb{R}^2) \bigcap
L^{\infty}(\mathbb{R}^2) \bigcap \mathcal{C}_{\xi}^1(\mathbb{R}^2)$)
such that $A \ast u_0 \rightarrow 0$ when $\mid \xi \mid \rightarrow
+\infty$.

In a second step, the nonlinearity is treated with the Duhamel
formulation. The solution is then the convolution of the group of the
linear equation with the nonlinear term $u \, u_{\xi}$ and all what is
already proved for the linear equation extends to the nonlinear one.

One may wonder whether such an interesting result can extend to the
viscous case. Indeed, the proof is even simpler and the details are
left to the interested reader. The Fourier transform of the
dissipative term writes
\[
\begin{array}{rl}
\frac{-1}{\sqrt{\pi R}}\mathcal{F}_{\xi \rightarrow X} \left( \Int_0^{+\infty} \Frac{\eta_{\xi}(\xi +\xi',y,\tau)}{\sqrt{\xi'}} \, {\rm d} \xi' \right)=& \frac{-1}{\sqrt{2\pi^2 R}} \Int_{\mathbb{R}_{\xi}} \Int_{\xi'=0}^{+\infty} \frac{\eta_{\xi }(\xi+\xi')}{\sqrt{\xi'}} e^{-i\xi X}\, {\rm d}\xi'{\rm d} \xi\\
 =& -\frac{i\, X}{\sqrt{\pi R}} \hat{\eta}(X)\Int_{\xi' =0}^{+\infty} \frac{e^{i \xi' X}}{\sqrt{\mid \xi' \mid}} {\rm d} \xi'.
\end{array}
\]
By a simple change of variable, one may see that the last integral is
indeed only a function of the sign of $X$ (denoted $\sgn (X)$) times
$1/\sqrt{\mid X \mid }$. This function of $\sgn (X)$ can be computed :
\[
\Int_{\xi=0}^{+\infty} \frac{e^{i \, \sgn(X) \xi}}{\sqrt{ \mid \xi \mid}}{\rm d} \xi=\sqrt{\frac{\pi}{2}}(1+i\sgn(X)),
\]
and the Fourier corresponding term reads:
\[
- i \, \hat{\eta}(X)\, \frac{X}{\sqrt{2R\mid X \mid}}(1+i\sgn(X)).
\]
As a consequence, one could study the fundamental solution
\[
G^{NS}(t,\xi,y)=\mathcal{F}^{-1}_{(X,y') \rightarrow (\xi,y)} \left[e^{it \left(\varepsilon X \mid X \mid^{\alpha}-y'^2/X+(1+i \sgn(X))X(2R| X |)^{-1/2}\right)}\right].
\]
The only real part (non-dispersive) inside the exponential is $it
\times i \sgn(X) X(2 R \mid X \mid)^{-1/2}=-t\mid X
\mid^{1/2}(2R)^{-1/2}$ which sign is compatible with the dissipation
and ensures convergence of all the integrals. So one do not even need
to use the Lebesgue dominated convergence theorem and the proof is
simpler than in \cite{MST_07}.

\section{Viscous Boussinesq equation with surface tension}
\label{sec.5}
The Boussinesq equation is a second order equation that takes into
account waves going both to the right and to the left. Since the
derivation is straightforward, we only sketch it, following Johnson
\cite{Johnson_97} (p. 216-219) and state the Proposition
\ref{prop_Boussinesq}. Johnson \cite{Johnson_97} uses a different
scaling, but only along the transverse direction which is not used
here.

Starting from the Boussinesq system (\ref{Boussinesq_1}) written only
in 1D (with $x,z,t$ and not $x,y,z,t$), one may differentiate
(\ref{Boussinesq_1})$_3$ with respect to $t$ and
(\ref{Boussinesq_1})$_1$ with respect to $x$. The difference of these
two equations writes after some easy computations:
\begin{multline}
\label{EB.2}
\eta_{tt}-\eta_{xx}-\varepsilon \left( u^2+\frac{\eta^2}{2} \right)_{xx}-\frac{\varepsilon}{3} \eta_{xxtt} + \frac{\eta_{xxxx}}{\Bo}+\frac{\varepsilon}{\sqrt{\pi R }} \eta_{tt} \ast \frac{1}{\sqrt{t}} \\
+\frac{\varepsilon }{\sqrt{\pi R}} {\rm p.v. }\Frac{1}{\sqrt{t}} \left[ u_x^{u,0}(x,z=0)-\Int_{\gamma'=0}^{+\infty} u^{b,0}_{x \gamma}(x,\gamma')\; e^{-\frac{R}{4t}\gamma'^2} {\rm d}\gamma'\right] = O(\varepsilon^2),
\end{multline}
where the convolution is in time and {\rm p.v.} denotes the principal
value as defined in the theory of distributions. The {\rm p.v.} stems
from an integration by parts in which one has a boundary integral
\[
\Frac{\varepsilon}{\sqrt{\pi R t}}\left[ f^0_{x}(x,\gamma')e^{-\frac{R \gamma'^2}{4t}} \right]_{\gamma'=0}^{+\infty}.
\]
The $u^{u,0}$ and $u^{b,0}$ are the initial values of $u^u$ and
$u^b$. They must be provided. Appart from the two viscous terms and
the term generated by surface tension, this equation is identical to
(3.41) of \cite{Johnson_97}. Since we concentrate on $x=O(1)$, we may
write $u=\int_{-\infty}^x u_x=-\int_{-\infty}^x \eta_t(x',t) {\rm
  d}x'+O(\varepsilon)$. The above equation is then the Eulerian form
of the Boussinesq equation rewritten in (\ref{Boussinesq_E}).\\

Let us now come back to the Lagrangian coordinates and change of
fields. In \cite{Johnson_97}, the author proposes $X =x+\varepsilon
\int_{-\infty}^x \eta(x',t) {\rm d}x'$, but we consider the following
change of variable easier to justify:
\begin{equation}
\label{EB.3}
\begin{array}{rl}
X& =x-\varepsilon \int_{0}^t u(x,t') {\rm d}t'\\
H(X,t)& =H(x-\varepsilon \int_{0}^t u(x,t') {\rm d}t',t) = \eta(x,t)-\varepsilon \eta^2(x,t).
\end{array}
\end{equation}
The function $H$ is assumed to be defined on $X$ in the initial
domain.  Using these definitions, it is straightforward to rewrite
(\ref{EB.2}) under the form of a Lagrangian Boussinesq equation with
surface tension and viscosity taken into account. Only the term
containing the initial horizontal velocities needs some insight. It is
the last one in (\ref{EB.2}) and both $u^{u,0}(x,z=0)$ and $u^{b,0}_{x
  \gamma}(x,\gamma')$ can be considered as depending on $x$ or on $X$
since an $\varepsilon$ appears in front of the whole term and the
difference between $x$ and $X$ is $O(\varepsilon)$. The following
Proposition states the two forms of the Boussinesq equation in 1D.
\begin{proposition}
\label{prop_Boussinesq}
Let a flow of a viscous fluid with surface tension in a 2D channel
with a flat bottom like the one depicted in Figure \ref{fig1}. Let
$u^{u,0}(x,z=0)$ denote the initial horizontal velocity in the upper
part (respectively $u^{b,0}(x,\gamma)$ in the boundary layer). We
assume the Boussinesq approximation and $1/\Bo =O(\varepsilon)$. Under
these assumptions, the surface waves obey the viscous Boussinesq
equation in the Eulerian form:
\begin{multline}
\label{Boussinesq_E}
\eta_{tt}-\eta_{xx} -\varepsilon \left( u^2+\eta^2/2 \right)_{xx}-\left(\frac{\varepsilon}{3} -\frac{1}{\Bo} \right) \eta_{xxxx}+\Frac{\varepsilon}{\sqrt{\pi R}} \eta_{tt} \ast \Frac{1}{\sqrt{t}}\\
+\Frac{\varepsilon}{ \sqrt{\pi R }}{\rm p.v. }\Frac{1}{\sqrt{t}} \left( -(\eta_t)_{t=0}-\Int_{\gamma'=0}^{+\infty}  u^{b,0}_{x \gamma}(x,\gamma')\; e^{-\frac{R\gamma'^2}{4t}} {\rm d} \gamma'\right) =O(\varepsilon^2),
\end{multline}
where the convolution is in time and {\rm p.v.} denotes the principal
value as defined in the theory of distributions. The Lagrangian form
(with (\ref{EB.3})) writes:
\begin{multline}
\label{Boussinesq_L}
H_{tt}(X,t)-H_{XX}-\Frac{3}{2}\varepsilon(H^2)_{XX}-\left(\frac{\varepsilon}{3}-\frac{1}{\Bo} \right) H_{XXXX}+\Frac{\varepsilon}{\sqrt{\pi R}} H_{tt} \ast \Frac{1}{\sqrt{t}}\\
+\Frac{\varepsilon }{\sqrt{\pi R}} {\rm p.v. }\Frac{1}{\sqrt{t}} \left[ -(H_t(X,t))_{t=0}-\Int_{\gamma'=0}^{+\infty} u^{b,0}_{x \gamma}(x,\gamma')\; e^{-\frac{R}{4t}\gamma'^2} {\rm d}\gamma'\right] =O(\varepsilon^2),
\end{multline}
where the initial velocities defined on $x$ can be evaluated either on
$x$ or on $X$.
\end{proposition}
Notice that the formulae (\ref{Boussinesq_E},\ref{Boussinesq_L}) also
have a formulation in terms of $u^{u,0}(x,z=0)-u^{b,0}(x,\gamma)$ which
is physically meaningful.

\section{Is the Boussinesq system consistent with Navier-Stokes ?}
\label{sec.6}

The weakly transverse Boussinesq system is stated in Proposition
\ref{Prop.1}. We intend to justify here that it is consistent with the
initial Navier-Stokes equations (\ref{Ga.5}).

What is consistency ? In \cite{LeMeur_14}, the author justifies, but
does not prove in any functionnal space, that if there is a solution
to Navier-Stokes equations, with a free boundary and a flat bottom,
then, under the Bousinesq approximation, the solution satisfies the
viscous isotropic Boussinesq system. Such a necessary result deserves
to be completed by a sufficient one which is consistency. This is the
goal of the present section.

Notice that the consistency as defined by Lannes in \cite{Lannes_2013}
(Definition 5.1) is much more rigorous. While we use only expansions,
Lannes asked whether a solution to the asymptotic model (assuming it
exists), in the asymptotic regime, satisfied the initial complete
model on a convenient time-interval and in convenient norms, up to a
given power of the small parameter. Lannes proved in \cite{Lannes_02}
that KP solution (with his scaling) is consistent with the Boussinesq
system.

We prove below the following Proposition in the 1D case.

\begin{proposition}
\label{consistency}
Let $u^u(\mathbf{x},z,t)$ be the horizontal velocity in the upper part
of a flow in a geometry as in Figure \ref{fig1}, and
$\eta(\mathbf{x},t)$ be the height of the free boundary above the
fluid, for $\mathbf{x}=(x,y) \in \mathbb{R}^2, z \in
(\varepsilon,1+\varepsilon \eta(\mathbf{x},t))$ and nonnegative time $t$. Let
$(u^{u,0}(\mathbf{x},z),v^{b,0}(\mathbf{x},z))$ denote the initial
horizontal velocity in the upper part (respectively
$(u^{b,0}(\mathbf{x},\gamma),v^{b,0}(\mathbf{x},\gamma))$ in the
boundary layer with $\gamma=z/\varepsilon$).

If $(u^u,v^u,\eta)$ satisfies the non-isotropic Boussinesq system
(\ref{Boussinesq_1}), then there exist fields in the upper part
$(u^u,v^u,w^u,\eta)$ (resp. $(u^b,v^b,w^b,\eta)$) satisfying the
Navier-Stokes equations (\ref{Ga.5}) in the upper part (resp. the
boundary layer) of the domain and the interface continuity at least as
$O(\varepsilon)$.
\end{proposition}

\begin{proof}
We start from the final Boussinesq model (\ref{Boussinesq_1}) written in 1D:
\begin{equation}
\label{eq83_84}
\begin{array}{rl}
u_t^u+\eta_x -\Frac{\eta_{xxx}}{\Bo}+\varepsilon u^u u_x^u-\varepsilon \eta_{xtt}\Frac{(z^2-1)}{2} & =  O(\varepsilon^2)+O(\varepsilon/\Bo),\\
\eta_t+u_x^u(x,z,t)-\Frac{\varepsilon}{2} \eta_{xxt} (z^2-\frac{1}{3})+ \varepsilon (u^u \eta)_x -\Frac{\varepsilon}{\sqrt{\pi R}} u_x^u \ast \Frac{1}{\sqrt{t}} &  \\
\hspace*{7mm}+\Frac{2\varepsilon}{\sqrt{\pi}}\Int_{\gamma''=0}^{+\infty} \left( u^{b,0}_x(x,\gamma'') -u^{u,0}_x(x,z=0) \right) \Int_{\gamma'=0}^{\sqrt{\frac{\Rb}{4t}}\gamma''} e^{ -\gamma'^2}{\rm d}\gamma'{\rm d}\gamma'' &=  o(\varepsilon),
\end{array}
\end{equation}
where $u^u$ is the horizontal velocity in the upper part ($z \in
(\varepsilon,1+\varepsilon \eta(x,t))$), $\eta(x,t)$ is  the free boundary's
height, $u^{u,0}(x,z)$ (resp. $u^{b,0}(x,\gamma)$) is the initial
horizontal velocity in the upper part (resp. in the boundary
layer). In the boundary layer, we set $z=\varepsilon \gamma$.

First, one must justify that $\eta_{tt}=\eta_{xx}+O(\varepsilon)$,
which is easy from the zeroth order of (\ref{eq83_84})$_1$. Then, the equation (\ref{eq83_84})$_1$ enables to prove
the following Lemma.
\begin{lemma}
A localized solution of (\ref{eq83_84})$_1$ is such that
\[
\begin{array}{rl}
 \int_0^1 u = & u(x,z,t)-\varepsilon \eta_{xt}(z^2-1/3)/2+O(\varepsilon^2),\\
u(x,0,t) = & u(x,z,t) - \varepsilon \eta_{xt} z^2 /2+O(\varepsilon^2),\\
u(x,1,t) = & u(x,z,t)+ \varepsilon \eta_{xt} (1-z^2)/2+O(\varepsilon^2).
\end{array}
\]
\end{lemma}
The way to prove this Lemma is identical to the proof of Lemma 11 of
\cite{LeMeur_14}. One differentiates (\ref{eq83_84})$_1$ with respect
to $z$, integrate with respect to time $t$, prove that the constants
of integration vanish if the wave is localized, and use that
$\eta_{tt}=\eta_{xx}+O(\varepsilon)$ to get $u_z=\varepsilon
\eta_{xt}z +O(\varepsilon^2)$. Completing the proof is then easy.

Thanks to the previous Lemma, if we define
\[
\begin{array}{rl}
p^u(x,z,t)=&(\eta-\eta_{xx}/\Bo)-\varepsilon \eta_{tt}(z-1)+\varepsilon \int_1^z \int_1^{z'} u^u_{xt}(x,z'',t){\rm d}z''\, {\rm d}z'\\
w^u(x,z,t)=&\varepsilon (\eta_t+\varepsilon \eta_x u^u(z=1+\varepsilon \eta))-\varepsilon\int_{1+\varepsilon \eta}^z u_x^u,
\end{array}
\]
then $(u^u,w^u,p^u)$ satisfies (\ref{Ga.5})$_{1,3,5,6}$ up to
$O(\varepsilon^2)+O(\varepsilon/ \Bo)$, and (\ref{Ga.5})$_{4}$ exactly
in the upper part.

We must now check (\ref{Ga.5}) in the boundary layer, where we define
$z=\varepsilon \gamma$ and
\[
\begin{array}{rl}
u^b(x,\gamma,t)= & u^u(x,z=0,t)+\frac{\srb}{2}\int_{0}^{+\infty} f_{0}(x,\gamma') \frac{e^{-\frac{R(\gamma'-\gamma)^2}{4t}}}{\sqrt{\pi t}} {\rm d} \gamma'\\
  & -u^u(x,0,.) \ast \mathcal{L}^{-1}_{ p \rightarrow t}(e^{-\sqrt{R p} \gamma}) -\frac{\srb}{2}\int_{0}^{+\infty} f_{0}(x,\gamma') \frac{e^{\frac{-R(\gamma'+\gamma)^2}{4t}}}{\sqrt{\pi t}} {\rm d} \gamma',
\end{array}
\]
where $f_0(x,\gamma)=u^{b,0}(x,\gamma)-u^{u,0}(x,z=0)$. It is proved
in Lemma 6 of \cite{LeMeur_14} that this function satisfies
\[
\begin{array}{rrl}
& (u^b-u^u(z=0))_t-(u^b-u^u(z=0))_{\gamma \gamma}/R &=0,\\
\Rightarrow &u_t^b+\eta_x-u^b_{\gamma \gamma}/R -\eta_{xxx}/ \Bo &= O(\varepsilon)+O(\varepsilon / \Bo).
\end{array}
\]
This $u^b$ also satisfies the initial conditions and the boundary
conditions:
\[
\left\{ \begin{array}{rl}
u^b(x,\gamma,t=0)= & u^{b,0}(x,\gamma),\\
u^b(x,\gamma=0,t) = & 0,\\
u^b(x,\gamma \rightarrow +\infty,t)= &u^{u}(x,z=0,t) \mbox{ (continuity condition)}.
\end{array}\right.
\]

The function $u^b$ is defined as $u^b-u^u(z=0)$ being the solution to
the heat equation. Because it is only the zeroth order of the
conservation of momentum in the boundary layer (\ref{Ga.5})$_1$, and
also because the limit between the boundary layer and the upper part
is not made sufficiently precise through the lift ``function''
$u^u(x,z=0,t)$, one may not prove more.

If we define:
\[
\begin{array}{rl}
p^b(x,\gamma,t) &= \eta(x,t)-\eta_{xx}/ \Bo+\varepsilon \eta_{xx}/2\\
w^b(x,\gamma,t) &= -\varepsilon^2 \int_0^{\gamma} u_x^b(x,\gamma',t)\, {\rm d} \gamma',
\end{array}
\]
then (\ref{Ga.5})$_1$ in the boundary layer reduces to:
\[
\begin{array}{l}
(u^b-u^u)_t-(u^b-u^u)_{\gamma \gamma}/R-\varepsilon u^u u^u_x+\varepsilon u^b u^b_x-\varepsilon u^b_{\gamma} \int_0^{\gamma} u^b_x+\varepsilon \eta_{xxx}/2+O(\varepsilon^2)+O(\varepsilon / \Bo),\\
=-\varepsilon u^u u^u_x+\varepsilon u^b u^b_x-\varepsilon u^b_{\gamma} \int_0^{\gamma} u^b_x+\varepsilon \eta_{xxx}/2+O(\varepsilon^2)+O(\varepsilon / \Bo) =O(\varepsilon),
\end{array}
\]
where $u^u$ is evaluated at $z=0$. So $(u^b,w^b,p^b)$ satisfies
(\ref{Ga.5})$_{1}$ only up to $O(\varepsilon)$. So as to go further,
one should use correctors, solve a nonlinear heat equation, make a
more precise study at the interface between the two subdomains and use
functionnal spaces. Simple computations prove that $(u^b,w^b,p^b)$
satisfies (\ref{Ga.5})$_{3,4,7}$ up to
$O(\varepsilon^2)+O(\varepsilon/ \Bo)$.

The overlapping of the boundary layer and the upper part makes that
consistency cannot be justified more precisely than
$O(\varepsilon)$. Yet, since the interface condition (for instance at
$z=\sqrt{\varepsilon}$) between these two domains is satisfied, the
proof is complete, up to order $O(\varepsilon)$.

\end{proof}

\section{Conclusion}

We modeled the water waves flow by the Navier-Stokes equations in a 3D
geometry with a flat bottom and free surface. We assumed the order of
magnitude of the velocity in the transverse direction and deduced the
scaling in this coordinate and discussed its influence on the results
(Remark \ref{rem_Lannes_Johnson}). The linear theory gave a new
dispersion relation and a new phase velocity (Proposition
\ref{prop1}). Surface tension was discussed and seemed often more
relevant than viscosity, although these two parameters act very
differently. We stated the associated Boussinesq system. This
intermediate system enabled us to derive the KdV equation (Proposition
\ref{prop_KdV}) and the KP equation or system (Proposition
\ref{prop_KP}). We justified why the zero-mass constraint does not
raise any trouble for viscous KP, using \cite{MST_07}. Using the
previous computations, we derived the Boussinesq equation (Proposition
\ref{prop_Boussinesq}). Reciprocally, we proved that, should they
exist, the solutions to the non-isotropic Boussinesq system are
consistent with the Navier-Stokes equations (Proposition
\ref{consistency}). These systems and equations are derived for a
viscous fluid and take the surface tension into account.





\end{document}